\documentclass{degruyter-journal-a}      

\usepackage{booktabs,tikz,pgfplots,algorithm,algpseudocode}

\pgfkeys{
    /pgf/number format/.cd,
        set thousands separator={\,},
        min exponent for 1000 sep=4,
}
 
\newcommand{\vc}[1]{\boldsymbol{#1}}

\newcommand{\eps}{\varepsilon}

\newcommand{\vol}{\operatorname{vol}}

\newcommand{\cS}{\mathcal{S}}
\newcommand{\cB}{\mathcal{B}}
\newcommand{\cL}{\mathcal{L}}
\newcommand{\cC}{\mathcal{C}}
\newcommand{\cH}{\mathcal{H}}
\newcommand{\cV}{\mathcal{V}}
\newcommand{\cK}{\mathcal{K}}
\newcommand{\cR}{\mathcal{R}}

\newcommand{\Sp}{\mathrm{S}}

\newcommand{\T}{\mathrm{T}}

\theoremstyle{definition}\newtheorem{heuristic}[theorem]{Heuristic assumption}

\title{Approximate Voronoi cells for lattices, revisited}

\author{Thijs Laarhoven}\lastnameone{Laarhoven}
\firstnameone{Thijs}
\nameshortone{T. Laarhoven}
\addressone{Eindhoven University of Technology, Eindhoven}
\countryone{The Netherlands}
\emailone{mail@thijs.com}

\abstract{We revisit the approximate Voronoi cells approach for solving the closest vector problem with preprocessing (CVPP) on high-dimensional lattices, and settle the open problem of Doulgerakis--Laarhoven--De Weger [PQCrypto, 2019] of determining exact asymptotics on the volume of these Voronoi cells under the Gaussian heuristic. As a result, we obtain improved upper bounds on the time complexity of the randomized iterative slicer when using less than $2^{0.076d + o(d)}$ memory, and we show how to obtain time--memory trade-offs even when using less than $2^{0.048d + o(d)}$ memory. We also settle the open problem of obtaining a continuous trade-off between the size of the advice and the query time complexity, as the time complexity with subexponential advice in our approach scales as $d^{d/2 + o(d)}$, matching worst-case enumeration bounds, and achieving the same asymptotic scaling as average-case enumeration algorithms for the closest vector problem.}



\researchsupported{The author is supported by a Veni grant from NWO under project number 016.Veni.192.005.}

\acknowledgments{The author thanks L\'{e}o Ducas for insightful discussions on the topic of approximate Voronoi cells.}

\begin{document}


\section{Introduction}

Ever since the discovery of polynomial-time quantum attacks on widely deployed public-key cryptosystems~\cite{shor94}, researchers have been looking for ways to construct cryptographic schemes whose security relies on problems which remain hard even when large-scale quantum computers become a reality~\cite{bernstein09, nist17, etsi19}. A prominent class of potentially ``post-quantum'' cryptosystems~\cite{ajtai97, regev05, stehle09} relies on the hardness of lattice problems, such as the shortest (SVP) and closest vector problems (CVP). Understanding their hardness is essential for an efficient and reliable deployment of lattice-based cryptographic schemes in practice.  

Over time, the practical hardness of SVP and CVP has been quite well studied, with two classes of algorithms emerging as the most competitive: \textit{enumeration}~\cite{kannan83, fincke85, gama10, micciancio15, aono17, aono18}, running in superexponential time $2^{\Theta(d \log d)}$ in the lattice dimension $d$ (the main security parameter), using a negligible amount of space; and \textit{sieving}~\cite{ajtai01, nguyen08, micciancio10b, laarhoven15crypto, ducas18, herold18, albrecht19}, running in only exponential time $2^{\Theta(d)}$, but also requiring an amount of memory scaling as $2^{\Theta(d)}$. The best asymptotic time complexities for enumeration ($d^{d/2e + o(d)}$ for SVP, $d^{d/2 + o(d)}$ for CVP~\cite{hanrot07}) and sieving ($(3/2)^{d/2 + o(d)}$ for both SVP and CVP~\cite{becker16lsf, laarhoven16sac}) have remained unchanged since 2007 and 2016 respectively,\footnote{This statement concerns classical complexities; for quantum complexities, see e.g.~\cite{laarhoven15dcc, aono18}.} and recent work has mainly focused on decreasing second-order terms in the time and space complexities~\cite{gama10, aono17, laarhoven18pqcrypto, ducas18, albrecht19}.

A close relative to CVP, the closest vector problem with preprocessing (CVPP), has received far less attention~\cite{micciancio01e, aharonov04, bonifas15, stephens19} -- from a practical point of view, only a few recent works have studied how preprocessing can be used to speed up CVP~\cite{laarhoven16sac, doulgerakis19}. Since a fast CVPP algorithm would imply faster lattice enumeration algorithms for SVP/CVP~\cite{gama10, laarhoven16sac, doulgerakis19}, faster approximate-SVP algorithms for ideal lattices~\cite{pellet19, stephens19}, and even faster isogeny-based cryptography~\cite{beullens19}, a better understanding of the hardness of CVPP is needed.

\subsection{Approximate Voronoi cells} 

A natural approach for solving nearest-point queries for large data sets is to use \textit{Voronoi cells}; partitioning the space in regions, where each cell contains all points closer to the point in this cell than to any other point in the data set. Micciancio--Voulgaris~\cite{micciancio10} proposed an algorithm for constructing the Voronoi cell $\cV$ of a lattice in time $2^{2d + o(d)}$ and space $2^{d + o(d)}$, which can then be used to solve CVPP in time $2^{2d + o(d)}$. Bonifas--Dadush~\cite{bonifas15} later improved the query time complexity to only $2^{d + o(d)}$, but with the best heuristic algorithms for CVP running in time and space less than $2^{0.3d + o(d)}$, using exact Voronoi cells seems impractical. 

To make the Voronoi cells approach practical, Laarhoven~\cite{laarhoven16sac} and Doulgerakis--Laarhoven--De Weger (DLW)~\cite{doulgerakis19} proposed constructing \textit{approximate Voronoi cells} of the lattice, and using a randomized version of the iterative slicer algorithm of Sommer--Feder--Shalvi~\cite{sommer09} for solving CVP queries. These cells $\cV_L$, defined by a list of lattice vectors $L \subset \cL$, can be seen as rough, low-memory approximations to the exact Voronoi cell $\cV$ -- low-quality representations of the same object, which attempt to model the object as well as possible within the limited space available. These approximate representations are lossy, but are also smaller and easier to store (less memory) and faster to process (less time).

For analyzing the performance of this approach, DLW conjectured a relation between the performance of the algorithm and how well $\cV_L$ approximates $\cV$:
\begin{align}
p = \Pr(\text{the iterative slicer, with input $L$, solves CVP}) \stackrel{?}{\approx} \frac{\vol(\cV)}{\vol(\cV_L)} \, .
\end{align} 
They then obtained upper bounds on the volume of $\cV_L$ relative to $\cV$ by studying the success probability of the randomized slicer. An open problem from DLW was to better study the volumes of these approximate Voronoi cells, as this may lead to tighter bounds on their CVPP algorithm. Furthermore, the time--space trade-offs from DLW seemed somewhat unnatural --- the query time complexity diverges when the memory is less than $2^{0.05d + o(d)}$ --- and a second open problem was to obtain time complexities scaling as $2^{\Theta(d)}$ for arbitrary memory complexities $2^{\Omega(d)}$.

\subsection{Volumes of approximate Voronoi cells}

In this paper we take a fundamental approach to studying the shape of approximate Voronoi cells. We model this problem as estimating the volume of the intersection of a large number of random half-spaces, and we solve the latter problem exactly for the main regimes of interest. In particular, without any heuristic assumptions, we prove the following result regarding the volume of a random polytope obtained by intersecting a large number of random half-spaces. Assuming that the distribution of lattice points inside a large ball can be approximated well by a uniform distribution over the ball, this then leads to a tight asymptotic estimate of the volume of approximate Voronoi cells.
\begin{theorem}[Volume of approximate Voronoi cells] \label{cor:lattice}
Let $\alpha > 1$, and let $L \subset \cL \setminus \{\vc{0}\}$ consist of the $\alpha^d$ shortest non-zero vectors of a lattice $\cL$. Then, assuming the Gaussian heuristic holds, with probability $1 - o(1)$ we have:
\begin{align}
\alpha \leq \sqrt{2} \quad \implies \quad \vol(\cV_L) &= \left(\frac{\alpha^4}{4 \alpha^2 - 4}\right)^{d/2 + o(d)} \vol(\cV); \label{eq:lattice1} \\
\alpha \geq \sqrt{2} \quad \implies \quad \vol(\cV_L) &= (1 + o(1))^{d + o(d)} \vol(\cV). \label{eq:lattice2}
\end{align}
\end{theorem}
Assuming~\cite[Heuristic assumption 1]{doulgerakis19} holds (which has been restated here as Heuristic~\ref{heur:rand}), this result would then imply what are the exact asymptotic time and space complexities of the randomized slicer. However, under the same assumption, DLW derived the following asymptotic upper bound on the relative volume of approximate Voronoi cells, for $\alpha \in (1, \sqrt{2})$:
\begin{align}
\frac{\vol(\cV_L)}{\vol(\cV)} \stackrel{?}{\leq} \left(\frac{16 \alpha^4 \left(\alpha^2 - 1\right)}{-9 \alpha^8 + 64 \alpha^6 - 104 \alpha^4 + 64 \alpha^2 - 16}\right)^{d/2 + o(d)}. \label{eq:dlw}
\end{align}
Looking closely, \eqref{eq:lattice1} in fact \textit{contradicts} the above upper bound for $\alpha > \frac{1}{3} \sqrt{10} \approx 1.054$. The source of this contradiction can be found in~\cite[Heuristic assumption 1]{doulgerakis19}: while this assumption states that the success probability $p$ of the randomized slicer is \textit{exactly} $p = \vol(\cV)/\vol(\cV_L)$, the randomized slicer is in fact \textit{more likely} to converge to short solutions than to long solutions: we may well have $p \gg \vol(\cV)/\vol(\cV_L)$, and the gap between both quantities may be exponentially large in $d$. A lower bound on $p$ therefore does not necessarily translate to a lower bound on $\vol(\cV)/\vol(\cV_L)$, or to an upper bound on its reciprocal.

\subsection{Application to CVPP}
 
Although \eqref{eq:dlw} is incorrect as an upper bound on the volume of approximate Voronoi cells, on closer inspection we see that to bound the complexity of their algorithm, DLW in fact proved that $p$ is at most the RHS of \eqref{eq:dlw}: the bound on the volume of the approximate Voronoi cell was then only obtained through transitivity by applying \cite[Heuristic assumption 1]{doulgerakis19}. Thus, letting $p_{\alpha}$ denote the success probability of the randomized slicer when using a list of the $n = \alpha^d$ shortest non-zero vectors in the lattice, we now have two heuristic lower bounds on $p_{\alpha}$:
\begin{align}
\text{(DLW)} \qquad p_{\alpha} &\geq \left(\frac{-9 \alpha^8 + 64 \alpha^6 - 104 \alpha^4 + 64 \alpha^2 - 16}{16 \alpha^4 \left(\alpha^2 - 1\right)}\right)^{d/2 + o(d)}; \label{eq:dlw111} \\
\text{(ours)} \qquad p_{\alpha} &\geq \left(\frac{4 \alpha^2 - 4}{\alpha^4}\right)^{d/2 + o(d)}.
\end{align}
These bounds are both conditional on the Gaussian heuristic, and the second result holds conditional on $p_{\alpha} \geq \vol(\cV) / \vol(\cV_L)$. By applying similar techniques from~\cite{doulgerakis19}, we obtain the following CVPP complexities, where $\delta = \sqrt{\alpha^2 - 1} / \alpha$.
\begin{theorem}[CVPP complexities] \label{thm:cvpp}
Let $\alpha \in (1, \sqrt{2})$ and $u \in (\delta, \frac{1}{\delta})$. Then we can heuristically solve CVPP with query space and time $\Sp$ and $\T$, where:
\begin{align}
 \Sp &= \left(\frac{\alpha}{\alpha - (\alpha^2 - 1) (\alpha u^2 - 2 u \sqrt{\alpha^2 - 1} + \alpha)}\right)^{d/2 + o(d)}, \\
 \T &= \left(\frac{\alpha^4}{4 \alpha^2 - 4} \cdot \frac{\alpha + u \sqrt{\alpha^2 - 1}}{{-\alpha^3 + \alpha^2 u \sqrt{\alpha^2 - 1} + 2 \alpha}}\right)^{d/2 + o(d)}.
\end{align}
The best query complexities $(\Sp, \T)$ together form the blue curve in Figure~\ref{fig:to}.
\end{theorem}
%

As we can see in Figure~\ref{fig:to}, for the low-memory regime of less than $2^{0.076d + o(d)}$ memory, we obtain strictly better query time complexities than~\cite{doulgerakis19}. The trade-offs from~\cite{doulgerakis19} were further limited to the regime of using at least $2^{0.048d + o(d)}$ memory, whereas Theorem~\ref{thm:cvpp} describes a continuous trade-off between the query time and space complexities: for arbitrary memory complexities $2^{\eps d + o(d)}$ with $\eps > 0$, we obtain a query time complexity $2^{\Theta(d)}$. Extending Theorem~\ref{thm:cvpp} to the regime of $\alpha = 1 + o(1)$, we obtain the following result.
\begin{corollary}[Polynomial advice for CVPP]
Using $d^{\Theta(1)}$ memory, we can heuristically solve CVPP in time $d^{d/2 + o(d)}$.
\end{corollary} 
This matches the asymptotic worst-case time complexities for solving CVP with enumeration of Hanrot--Stehl\'{e}~\cite{hanrot07}, and with an average-case scaling for enumeration of $d^{d/(2e) + o(d)}$, this is only off by a factor $1/e$ in the exponent compared to practical enumeration methods. We further see that if we use a preprocessed list of size e.g.\ $2^{\Theta(d^{\gamma})}$ for constant $\gamma \in (0,1)$, we heuristically obtain a CVPP time complexity scaling as $2^{\frac{1}{2}(1 - \gamma) d \log_2 d + o(d \log d)}$.

\paragraph{Outline.} Section~\ref{sec:preliminaries} first defines notation and preliminary results. Section~\ref{sec:polytopes} studies the volume of intersections of random halfspaces. Section~\ref{sec:cvpp} describes the application of these results to solving CVPP and the resulting trade-offs. The appendices describe further details on prior work, to make the paper self-contained.

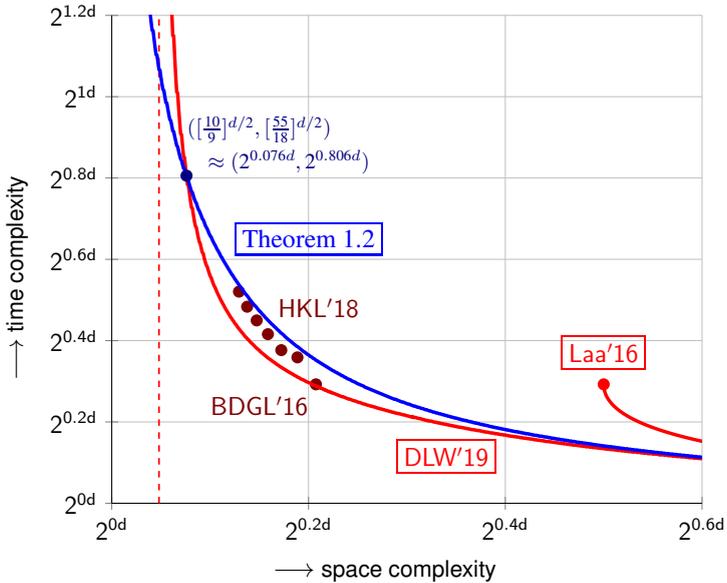
\begin{figure}[!t]
{\center
\pgfplotsset{every tick label/.append style={font=\large}}
\begin{tikzpicture}[scale=0.85]
\begin{axis}[
    xmin=0, 
    xmax=0.6,
    ymin=0, 
    ymax=1.2,
    grid=both,
    xtick={-0.2,0.0,0.2,...,1.2},
    ytick={-0.2,0.0,0.2,...,2.2},
    xscale=1.4,
    yscale=1.4,
    axis y line=left,
    axis x line=bottom,
    yticklabel={$\mathsf{2^{\pgfmathprintnumber{\tick}d}}$},
    xticklabel={$\mathsf{2^{\pgfmathprintnumber{\tick}d}}$},
    axis line style={-},
    x label style={at={(axis description cs:0.33,-0.025)},anchor=north},
    y label style={at={(axis description cs:0,0.33)},anchor=south},
    xlabel={\large \sffamily $\longrightarrow$ space complexity},
    ylabel={\large \sffamily $\longrightarrow$ time complexity}
    ]
    

\fill [red!50!black] (20.75, 29.25) ellipse (2pt and 2pt); 

\fill [red!50!black] (18.87, 35.88) ellipse (2pt and 2pt);
\fill [red!50!black] (17.23, 37.66) ellipse (2pt and 2pt);
\fill [red!50!black] (15.87, 41.59) ellipse (2pt and 2pt);
\fill [red!50!black] (14.73, 44.97) ellipse (2pt and 2pt);
\fill [red!50!black] (13.76, 48.34) ellipse (2pt and 2pt);
\fill [red!50!black] (12.93, 52.05) ellipse (2pt and 2pt);

\addplot[domain = sqrt(1/2) : sqrt(2), samples=100, red, ultra thick] 
	({log2(1 / (x * (sqrt(2) - x))) / 2}, 
	{log2((sqrt(2) + x) / (2 * x)) / 2}); 
\fill [red] (50, 29.25) ellipse (2pt and 2pt);

\addplot[domain = 0 : 2, samples=100, red, dashed, thick] 
	(0.048, x); 

\node [red!50!black] at (15,24.0) {\large $\mathsf{BDGL'16}$};
\node [red!50!black] at (21,48.0) {\large $\mathsf{HKL'18}$};

\node [red] at (50,37.0) {\large \fbox{$\mathsf{Laa'16}$}};

\node [red] at (34,11.5) {\large \fbox{$\mathsf{DLW'19}$}};
\node [blue] at (20,65.0) {\large \fbox{Theorem~\ref{thm:cvpp}}};

\pgfplotstableread{test.dat}
	\datatable

\addplot [color = red, ultra thick] table[y = Time] from \datatable;

\pgfplotstableread{CVP-new.dat}
	\tableee
	
\addplot [color = blue, ultra thick] table[y = Time] from \tableee;

\fill [blue!50!black] (7.60, 80.57) ellipse (2pt and 2pt);

\node [blue!50!black] at (15,92.0) {\small $([\frac{10}{9}]^{d/2}, [\frac{55}{18}]^{d/2})$};
\node [blue!50!black] at (18,84.0) {\small $\approx (2^{0.076d}, 2^{0.806d})$};

\end{axis}
\end{tikzpicture}}
\caption{Query complexities for solving CVPP. The labeled curves and points correspond to the papers~\cite{laarhoven16sac, becker16lsf, herold18, doulgerakis19}. Our new upper bound on the query time complexity improves upon DLW when using less than $(10/9)^{d/2 + o(d)} \approx 2^{0.076d + o(d)}$ memory. Note that, whereas the red DLW-curve diverges as the memory approaches the dashed asymptote $2^{0.048d + o(d)}$ from above, our trade-offs heuristically continue all the way to the regime of subexponential memory. \label{fig:to}}
\end{figure}

\clearpage


\section{Preliminaries}
\label{sec:preliminaries}

Given a set $\vc{B} = \{\vc{b}_1, \dots, \vc{b}_d\} \subset \mathbb{R}^d$ of linearly independent vectors, we define $\cL = \cL(\vc{B}) := \{\sum_{i=1}^d \lambda_i \vc{b}_i : \vc{\lambda} \in \mathbb{Z}^d\}$ as the lattice generated by $\vc{B}$. We write $\|\cdot\|$ for the Euclidean norm. Given a basis of a lattice and a target vector $\vc{t} \in \mathbb{R}^d$, the closest vector problem (CVP) is to find the vector $\vc{v} \in \cL$ closest to $\vc{t}$. In the preprocessing version (CVPP), the problem is split into two parts: the preprocessing phase (without knowing $\vc{t}$) and the query phase (with knowledge of $\vc{t}$). For CVPP, the task is to do preprocessing such that CVP queries can then be answered more efficiently than when solving CVP directly.

Let us define some basic high-dimensional objects below, where $\vc{v} \in \mathbb{R}^d$.
\begin{align}
\text{(unit sphere)} \qquad \quad \ \cS &:= \{\vc{x} \in \mathbb{R}^d: \|\vc{x}\| = 1\}, \\
\text{(unit ball)} \qquad \quad \ \cB &:= \{\vc{x} \in \mathbb{R}^d: \|\vc{x}\| \leq 1\}, \\
\text{(half-space)} \qquad \ \ \cH_{\vc{v}} &:= \{\vc{x} \in \mathbb{R}^d: \|\vc{x}\| \leq \|\vc{x} - \vc{v}\|\}, \\
\text{(convex polytope)} \qquad \ \ \ \cV_L &:= \bigcap_{\vc{v} \in L} \cH_{\vc{v}}, \qquad (\vc{0} \notin L) \\
\text{(spherical cap)} \qquad \ \ \ \cC_{\vc{v}} &:= \overline{\cH_{\vc{v}}} \cap \cB, \\
\text{(Voronoi cell)} \qquad \quad \ \cV &:= \cV_{\cL \setminus \{\vc{0}\}}.
\end{align}
We further define the complements $\overline{\cH_{\vc{v}}} := \mathbb{R}^d \setminus \cH_{\vc{v}}$ and $\overline{\cV_L} := \mathbb{R}^d \setminus \cV_L$ 
in $\mathbb{R}^d$, and $\overline{\cC_{\vc{v}}} := \cB \setminus \cC_{\vc{v}}$ 
on the ball. Note that the definition of a polytope $\cV_L$ is generic, and the list $L$ need not be from a lattice. $\cV_L$ may further be unbounded (and its volume may be infinite), although for sufficiently large, randomly chosen lists $L$ it will usually be finite. For $L \subset \cL \setminus \{\vc{0}\}$, the polytope $\cV_L$ defines an \textit{approximate Voronoi cell} of the lattice $\cL$~\cite{doulgerakis19}, satisfying $\cV \subseteq \cV_L$ with equality iff $\cR \subseteq L$, where $\cR$ is the set of \textit{relevant vectors} of the lattice~\cite{micciancio10}. 

To analyze volumes of intersections on the ball, we will use the following asymptotic formula~\cite[Equation (28)]{shannon59}, where $\alpha = \frac{1}{2} \|\vc{v}\| \in (0,1)$:
%
\begin{align}
C(\alpha) := \frac{\vol(\cC_{\vc{v}})}{\vol(\cB)} \sim \sqrt{\frac{1 - \alpha^2}{2 \pi \alpha^2 d}} \cdot (1 - \alpha^2)^{d/2}. \qquad \quad (d \to \infty) \label{eq:cap}
\end{align}
%
For constant $\alpha \in (0, 1)$ and large $d$, Equation~\eqref{eq:cap} can alternatively be written as $C(\alpha) = O((1 - \alpha^2)^{d/2} / \sqrt{d}) = (1 - \alpha^2)^{d/2 + o(d)}$. 

Finally, the \textit{Gaussian heuristic} states that for sufficiently smooth and random regions $\cK \subset \mathbb{R}^d$, the number of lattice points inside $\cK$ scales as $\vol(\cK) / \vol(\cV)$.

\clearpage


\section{Volumes of random polytopes}
\label{sec:polytopes}

To study the asymptotic behavior of volumes of approximate Voronoi cells, we will first study the more fundamental problem of estimating the volume of polytopes $\cV_L$ defined as the intersection of a large number of random half-spaces. We will study two specific cases for the list $L$ below:
\begin{enumerate}
\item[1.] Uniformly random points from the (unit) sphere;
\item[2.] Uniformly random points from the (unit) ball.
\end{enumerate}
The volume of such random polytopes has been previously studied in e.g.~\cite{shcherbina00, talagrand00, pivovarov07, turchi19, pivovarov10}, and in particular the case of points from the sphere was analyzed in~\cite{pivovarov07}. For the application to approximate Voronoi cells we need bounds for the case when points are drawn uniformly at random from a ball, which to the best of our knowledge has not been explicitly studied before. For completeness, and to illustrate how the analysis changes between the case of the sphere and the ball, we treat the case of random points from the unit sphere here as well. 

\subsection{Uniformly random points from the (unit) sphere}

First, let us study the case where $L$ is sampled uniformly at random from the unit sphere $\cS$. This setting was previously studied in \cite[Section 3.2]{pivovarov07}, but for extending the analysis to the case of the unit ball we explicitly analyze this problem here as well. Note that for $L \subseteq \cS^{d-1}$ we have the trivial lower bound $\vol(\cV_L) \geq 2^{-d} \vol(\cB)$, as $\frac{1}{2} \cB \subseteq \cV_L$. For a slightly less trivial upper bound, note that the polytope $\cV_L$ is unbounded iff all points in $L$ lie in a certain hemisphere. The probability that this happens was computed by Wendel~\cite{wendel62} as:
\begin{align}
\Pr_{L \sim \cS}\Big(\vol(\cV_L) < \infty\Big) = 1 - 2^{-n+1} \sum_{k=0}^{d-1} \binom{n-1}{k}.
\end{align}
In particular, it is extremely unlikely that for lists of size $n = \omega(d)$, the corresponding polytopes are unbounded. For lists of exponential size, we obtain the following result, similar to~\cite[Theorem 3.9]{pivovarov07}.

\begin{theorem}[Random points from the sphere] \label{thm:sphere}
Let $\alpha > 1$, and let $L \subset \cS$ consist of $n = \alpha^d$ uniformly random vectors from $\cS$. Then, with probability $1 - o(1)$ over the randomness of $L$, we have:
\begin{align}
\vol(\cV_L) = \left(\frac{\alpha^2}{4 \alpha^2 - 4}\right)^{d/2 + o(d)} \vol(\cB).
\end{align}
\end{theorem}

\begin{proof}
To prove Theorem~\ref{thm:sphere}, we will prove the following, equivalent statement:
\begin{align}
\vol(\cV_L) = \vol(r_0 \cB)^{1 + o(1)}, \qquad r_0 = \sqrt{\frac{\alpha^2}{4\alpha^2 - 4}} \, .
\end{align}
Note that $\vol(r \cB) = r^d \vol(\cB)$ for arbitrary $r$, hence the equivalence. Below we will further use the quantity $\cV_L^{(r)} = \cV_L \cap \, r \cB \subseteq \cV_L$ as the intersection of the polytope with the ball of radius $r > 0$. Observe that for sufficiently small $r \ll r_0$ we have $\cV_L^{(r)} = r \cB \subset \cV_L$ while for large $r \gg r_0$ we have $\cV_L^{(r)} = \cV_L \subset \, r \cB$. The quantity $r_0$ is intuitively the radius $r$ for which $\vol(\cV_L^{(r)}) \approx \vol(\cV_L) \approx \vol(r \cB)$.

First, some simple manipulations give:
\begin{align}
\cV_L^{(r)} = \bigcap_{\vc{v} \in L} \cH_{\vc{v}} \cap (r \cB) = \bigcap_{\vc{v} \in L} \left(r \cB \setminus r \cC_{\vc{v}/r}\right) = r \underbrace{\Big(\cB \setminus \bigcup_{\vc{v} \in L} \cC_{\vc{v}/r}\Big)}_{\cK}. \label{eq:vl}
\end{align} 
Note that the vectors $\vc{v}/r$ all have norm $1/r$, and the spherical caps $\cC_{\vc{v}/r}$ thus have a fixed base radius of $1/(2r)$. To prove the lower bound on $\vol(\cV_L)$, we will use elementary volume arguments to argue that $\vol(\cK) \approx \vol(\cB)$. For the upper bound, we have $\vol(\cK) \leq \vol(\cB)$, and we will argue that with high probability over the randomness of $L$, $\vol(\cV_L) \approx \vol(\cV_L^{(r)})$.

\textbf{Lower bound $(\geq)$}: Ignoring spherical cap intersections, we have:
\begin{align}
\vol(\cK) 
\geq \vol(\cB) - n \cdot \vol(\cC_{\vc{v}/r}) = \vol(\cB) \left[1 - \alpha^d \left(1 - \tfrac{1}{4 r^2}\right)^{d/2 + o(d)}\right].
\end{align}
For $1/\alpha^2 = 1 - 1/(4r^2) + o(1)$, or equivalently $r = r_0 - o(1)$, we thus get $\vol(\cK) \geq (1 - o(1)) \cdot \vol(\cB)$. 

\textbf{Upper bound $(\leq)$}: Clearly $\vol(\cV_L^{(r)}) = r^d \vol(\cK) \leq r^d \vol(\cB)$; the difficulty lies in showing that $\vol(\cV_L) \approx \vol(\cV_L^{(r)})$. Note that when $n$ is large, then the spherical caps in \eqref{eq:vl} will cover (almost) the entire surface of $\cB$ -- if e.g.\ only a fraction $2^{-\Theta(d^2)}$ of the sphere remains uncovered, then the parts of $\cV_L$ extending beyond $r \cB$ will contribute a negligible amount to the volume of $\cV_L$.

Given a point on $\cS$, the probability of it not being covered by one of $n$ spherical caps $\cC_{\vc{v}/r}$ is given by $[1 - \vol(\cC_{\vc{v}/r}) / \vol(\cB)]^n$. For $n = \vol(\cB) / \vol(\cC_{\vc{v}/r})$, this can be upper bounded by $1/e$, hence for $n = 2 d^2 \vol(\cB) / \vol(\cC_{\vc{v}/r})$ the expected quantity not covered on the sphere is at most $e^{-2 d^2}$. By Markov's inequality, the probability that more than a fraction $e^{-d^2}$ of the sphere is covered is at most $e^{-2d^2 + d^2} = e^{-d^2}$, and so the upper bound follows.
\end{proof}


\subsection{Uniformly random points from the (unit) ball}

As sampling from $\cB$ and $\cS$ is similar in high-dimensional spaces (almost all the volume of the ball is concentrated near the surface of the sphere), in most cases the asymptotics for the unit sphere and the unit ball are the same. However, when $n$ is very large, a significant number of vectors will have norm significantly less than $1$, and these will then determine the shape of the resulting polytope.

The following main result shows that if $n \gg 2^{d/2}$, then the volume of the Voronoi cell for $\vc{0}$ scales like $\vol(\cB)/n$. Note that $\cV_L$ can be seen as the Voronoi cell for $\vc{0}$ in the data set $L \cup \{\vc{0}\}$, and for $n \gg 2^{d/2}$ the Voronoi cell of the $\vc{0}$-vector is therefore no larger than the Voronoi cells of the other $n$ points in the ball -- each of the points covers an equal fraction $\vol(\cB)/n$ of the ball. For small $n$, the portion of the ball covered by $\vc{0}$ is an exponential factor larger than the average.

\begin{theorem}[Random points from the unit ball] \label{thm:ball}
Let $\alpha > 1$, and let $L \subset \cB$ consist of $n = \alpha^d$ uniformly random vectors from $\cB$. Then, with probability $1 - o(1)$ over the randomness of $L$, we have:
\begin{align}
\alpha \leq \sqrt{2} \quad \implies \quad \vol(\cV_L) &= \left(\frac{\alpha^2}{4 \alpha^2 - 4}\right)^{d/2 + o(d)} \vol(\cB); \\
\alpha \geq \sqrt{2} \quad \implies \quad \vol(\cV_L) &= \left(\frac{1}{\alpha^2}\right)^{d/2 + o(d)} \vol(\cB).
\end{align}
\end{theorem}

\begin{proof}
For $\gamma < 1$ close to $1$, let us divide the set $L$ into sets $L_i = \{\vc{v} \in L: \gamma^i \leq \|\vc{v}\| \leq \gamma^{i-1}\}$, for $i = 0, 1, \dots$, i.e.\ we partition $L_i$ according to a sequence of thin spherical shells. With high probability over the randomness of $L$, each of these lists $L_i$ will contain $(\gamma^i \alpha)^{d + o(d)}$ vectors. The original polytope can now equivalently be described as $\cV_L = \bigcap_{i=0}^{\infty} \cV_{L_i}$. To estimate the volume of $\cV_L$, note that by Theorem~\ref{thm:sphere}, each of these cells $\cV_{L_i}$ is roughly shaped like a ball of a certain radius $r_i$. As a result, the volume of $\cV_L$ is determined by the smallest radius $\min_{i \in \mathbb{N}} r_i$ of these balls, corresponding to one of the lists $L_i$. 

To find the list $L_i$ defining the smallest polytope, recall that by applying Theorem~\ref{thm:sphere} with $n_i = (\gamma^i \alpha)^{d + o(d)}$ vectors to a sphere of radius $\gamma^i$, we have the following relation, where $\beta = \gamma^{2i}$:
\begin{align}
\vol(\cV_{L_i}) = \left(\tfrac{\alpha^2 \gamma^{2i}}{4 \alpha^2 \gamma^{2i} - 4}\right)^{d/2 + o(d)} \vol(\gamma^i \cB) = \Big(\underbrace{\tfrac{\beta^2 \alpha^2}{4 \beta \alpha^2 - 4}}_{f(\beta)}\Big)^{d/2 + o(d)} \vol(\cB).
\end{align}
To find the value $\beta$ resulting in the smallest radius, note that the derivative of $f(\beta)$ satisfies $f'(\beta) = -\beta \alpha^2 (2 - \beta \alpha^2) / 4 (\beta \alpha^2 - 1)^2$, which is negative for small $\beta < 2/\alpha^2$, i.e.\ $f(\beta)$ is decreasing with $\beta$, and the volume of the $\cV_{L_i}$ increases with $i$. Now $f'(\beta) = 0$ has one solution at $\beta = 2 / \alpha^2$, which is attained by one of the lists $L_i$ iff $\alpha \geq \sqrt{2}$. In the regime $\alpha < \sqrt{2}$, the smallest radius is obtained for the first list $L_0$, resulting in the same bound as in Theorem~\ref{thm:sphere}, while for $\alpha \geq \sqrt{2}$ the non-trivial minimum value lies at $\beta = \gamma^{2i} = 2 / \alpha^2$, resulting in $f(\beta) = 1/\alpha^2$ and $\vol(\cV_L) = \alpha^{-d + o(d)} \vol(\cB)$. \qedhere
\end{proof}


Let us finally state separately what happens when we draw points uniformly at random from a ball of a different radius. This directly follows from Theorem~\ref{thm:ball}.

\begin{corollary}[Random points from the $\beta$-ball] \label{thm:ball2}
Let $\alpha > 1$, and let $L \subset \cB$ consist of $n = \alpha^d$ uniformly random vectors from $\beta \cdot \cB$. Then, with probability $1 - o(1)$ over the randomness of $L$, we have:
\begin{align}
\alpha \leq \sqrt{2} \quad \implies \quad \vol(\cV_L) &= \left(\frac{\alpha^2 \beta^2}{4 \alpha^2 - 4}\right)^{d/2 + o(d)} \vol(\cB); \\
\alpha \geq \sqrt{2} \quad \implies \quad \vol(\cV_L) &= \left(\frac{\beta^2}{\alpha^2}\right)^{d/2 + o(d)} \vol(\cB).
\end{align}
\end{corollary}

\begin{proof}
Relative to the $\beta$-ball, we have $\vol(\cV_L) = r^{d + o(d)} \vol(\beta \cB)$ with $r$ as in Theorem~\ref{thm:ball}. Noting that $\vol(\beta \cB) = \beta^d \vol(\cB)$, the result follows. \qedhere
\end{proof}



\section{Approximate Voronoi cells, revisited}
\label{sec:cvpp}

With the results from Section~\ref{sec:polytopes}, we can immediately deduce asymptotics for the volume of approximate Voronoi cells, where these results can now be derived using only the Gaussian heuristic, which has been used and verified on far more occasions than \cite[Heuristic 1]{doulgerakis19}.\footnote{By applying the Gaussian heuristic to balls of different radii, we can derive the density of norms of lattice vectors, while spherical symmetry of the distribution of lattice vectors then implies that the lattice vectors inside a ball must follow a uniform distribution.} 

\begin{corollary}[Points from a lattice] \label{cor:lattice2}
Let $\alpha > 1$, and let $L \subset \cL \setminus \{\vc{0}\}$ consist of the $\alpha^d$ shortest non-zero vectors of a lattice $\cL$. Then, assuming the Gaussian heuristic holds, with probability $1 - o(1)$ we have:
\begin{align}
\alpha \leq \sqrt{2} \quad \implies \quad \vol(\cV_L) &= \left(\frac{\alpha^4}{4 \alpha^2 - 4}\right)^{d/2 + o(d)} \vol(\cV); \\
\alpha \geq \sqrt{2} \quad \implies \quad \vol(\cV_L) &= (1 + o(1))^{d/2 + o(d)} \vol(\cV).
\end{align}
\end{corollary}

\begin{proof}
Without loss of generality, suppose that $\vol(\cV) = \vol(\cB)$. Under the Gaussian heuristic, the points $L$ are then essentially uniformly distributed in the ball of radius $\alpha$. Applying Corollary~\ref{thm:ball2} with $\alpha = \beta$, the result then follows. \qedhere
\end{proof}

\subsection{Heuristic assumptions}

Assuming that \cite[Heuristic assumption 1]{doulgerakis19} holds, as discussed in the introduction this would give us tight bounds on the success probability of the randomized iterative slicer from \cite{doulgerakis19}. However, these results would then contradict the claimed lower bound on the success probability from~\cite[Equation (37)]{doulgerakis19}. The source of this contradiction is \cite[Heuristic assumption 1]{doulgerakis19}, which reads as follows.\footnote{For details and definitions of $D_{\vc{t} + \cL, s}$ and $\operatorname{Slice}_L(\vc{t}')$, we refer the reader to \cite{doulgerakis19}.}
\begin{heuristic}[Randomized slicing, DLW] \label{heur:rand}
For $L \subset \cL$ and large $s$,
\begin{align}
\Pr_{\vc{t}' \sim D_{\vc{t} + \cL, s}}\Big[\operatorname{Slice}_L(\vc{t}') \in \cV\Big] \approx \frac{\vol(\cV)}{\vol(\cV_L)} \, . \label{eq:rand}
\end{align}
\end{heuristic}
In fact, the randomized slicer is biased towards finding as short solutions as possible, and the probability of returning the unique representative from $\cV$ may be much larger than $\vol(\cV)/\vol(\cV_L)$. We therefore propose using the following heuristic assumption instead:
\begin{heuristic}[Randomized slicing, new] \label{heur:rand2}
For $L \subset \cL$ and large $s$,
\begin{align}
\Pr_{\vc{t}' \sim D_{\vc{t} + \cL, s}}\Big[\operatorname{Slice}_L(\vc{t}') \in \cV\Big] \gtrsim \frac{\vol(\cV)}{\vol(\cV_L)} \, . \label{eq:rand2}
\end{align}
\end{heuristic}
To motivate this new assumption, consider the reverse process of starting at the sliced solution vector $\vc{t}'' = \operatorname{Slice}_L(\vc{t}')$, and adding lattice vectors of length at most $\alpha \lambda_1(\cL)$ to obtain longer and longer vectors in the coset $\vc{t} + \cL$. Now, given an initial sampled vector $\vc{t}' \sim D_{\vc{t} + \cL, s}$, the probability of reaching $\vc{t}''$ out of all possible solution vectors in $\vc{t} + \cL$ is essentially proportional to the number of paths from $\vc{t}''$ to $\vc{t}'$ through the above process of adding lattice vectors of length at most $\alpha \lambda_1(\cL)$ to $\vc{t}''$. Starting from a shorter vector, the tree of potential paths to $\vc{t}''$ is likely to be wider, and there are likely more such paths reaching $\vc{t}''$. 

Assuming that indeed, the success probability is \textit{at least} proportional to the ratio of these volumes, we obtain the CVPP complexities described in Theorem~\ref{thm:cvpp} in the introduction. Here we simply replaced the upper bound on $p_{\alpha}$ from~\cite{doulgerakis19} by the upper bound obtained via the volume of approximate Voronoi cells, and otherwise applied the same techniques of nearest neighbor speed-ups.



\subsection{The low-memory regime}

As Theorem~\ref{thm:cvpp} describes complexities even for the regime of $2^{\eps d + o(d)}$ memory with small $\eps$, let us study the asymptotic behavior as the memory is actually subexponential or even polynomial in $d$.

First, note that for the lower bound on the volume, we essentially only needed Equation~\eqref{eq:cap}, which holds even when $\alpha = o(1)$ scales with $d$. (See also~\cite[Lemmas 4.1 and 4.2]{pivovarov07} for absolute bounds.) For the upper bounds, we needed that the list $L$ properly covers the sphere, and we argued that $n = 2 d^2 \vol(\cB) / \vol(\cC_{\vc{v}})$ suffices to cover enough of the sphere with high probability. We can therefore extend these results all the way up to the regime of polynomial space. Note that for small $\alpha = 1 + \eps$, Theorem~\ref{thm:ball} gives:
\begin{align}
\vol(\cV_L) = \left(\frac{1}{\sqrt{8 \eps}} + O(\sqrt{\eps})\right)^{d + o(d)} \vol(\cB).
\end{align}
Substituting suitable values of $\alpha$, we get the following results.
\begin{proposition}[Polynomially many points from the unit ball] \label{thm:ball-eps}
Let $L \subset \cB$ consist of $n = d^{\Theta(1)}$ uniformly random vectors from $\cB$. Then, with probability $1 - o(1)$ over the randomness of $L$, we have $\vol(\cV_L) = 2^{\frac{1}{2} d \log_2 d + o(d \log d)} \vol(\cB)$.
\end{proposition}
\begin{proof}
This follows from substituting $\alpha = d^{\Theta(1/d)} = 1 + \Theta(\log d) / d$. \qedhere
\end{proof}
In the application of CVPP algorithms, Proposition~\ref{thm:ball-eps} shows that heuristically, we obtain a smooth trade-off between enumeration and using exact Voronoi cells -- Hanrot--Stehl\'{e}~\cite[Theorem 4]{hanrot07} previously showed that enumeration has a cost of $d^{d/2 + o(d)}$ time for solving CVP in the worst case, with polynomial memory.
\begin{proposition}[Subexponentially many points from the unit ball] \label{thm:ball-eps2}
Let $L \subset \cB$ consist of $n = 2^{\Theta(d^{\gamma})}$ uniformly random vectors from $\cB$. Then, with probability $1 - o(1)$ over the randomness of $L$, we have $\vol(\cV_L) = 2^{\frac{1}{2} (1 - \gamma) d \log_2 d + o(d \log d)} \vol(\cB)$.
\end{proposition}
\begin{proof}
This follows from substituting $\alpha = \exp \Theta(d^{\gamma-1}) = 1 + \Theta(d^{\gamma - 1})$.
\end{proof}
This matches results from e.g.\ \cite{dadush14}. To illustrate Proposition~\ref{thm:ball-eps2} with an example, we expect to be able to solve CVPP with query time $d^{d/4 + o(d)}$ when using $2^{\Theta(\sqrt{d})}$ memory, or we can match the average-case complexity of enumeration with a query time complexity of $d^{d/(2e) + o(d)}$ using $2^{\Theta(d^{1 - 1/e})} \approx 2^{\Theta(d^{0.63})}$ memory.

\providecommand{\bysame}{\leavevmode\hbox to3em{\hrulefill}\thinspace}
\providecommand{\MR}{\relax\ifhmode\unskip\space\fi MR }
\providecommand{\MRhref}[2]{%
  \href{http://www.ams.org/mathscinet-getitem?mr=#1}{#2}
}
\providecommand{\href}[2]{#2}

\appendix

\section{The Sommer--Feder--Shalvi iterative slicer}

We briefly describe some more details on previous, related work in these appendices, starting with the iterative slicer of Sommer--Feder--Shalvi~\cite{sommer09}. This algorithm provides an elementary, greedy strategy to attempt to find a closest vector to a given target vector $\vc{t}$, given a list of lattice points $L \subset \cL$, which always finds a solution when $L = \cR$ is the set of relevant vectors of the lattice. To do this, note that the shortest representative $\vc{t}'$ in the coset of the lattice $\vc{t} + \cL$ is necessarily contained in the Voronoi cell of the lattice, and therefore $\vc{0}$ is the closest lattice vector to $\vc{t}'$. This implies that $\vc{t} - \vc{t}'$ is the closest lattice vector to $\vc{t}$, and so finding the shortest representative $\vc{t}' \in \vc{t} + \cL$ is equivalent to solving CVP for $\vc{t}$.

To find this shortest representative, given $\vc{t}$ and a list of lattice vectors $L \subset \cL$, the algorithm follows the same approach of e.g.\ lattice sieving algorithms~\cite{nguyen08, micciancio10b, laarhoven16sac}: we start with $\vc{t}' = \vc{t}$, and we repeatedly try to find vectors $\vc{v} \in L$ such that $\vc{t}' \leftarrow \vc{t}' - \vc{v}$ is a shorter vector in the coset $\vc{t} + \cL$. If no more such reductions can be done, we terminate and hope that the algorithm found the shortest representative.

Summarizing, the iterative slicer can be succinctly described through the pseudocode of Algorithm~\ref{alg:slicer}.

\begin{algorithm}[!ht]
\caption{The Sommer--Feder--Shalvi iterative slicer~\cite{sommer09}}
\label{alg:slicer}
\begin{algorithmic}[1]
\Require The relevant vectors $\cR \subset \cL$ and a target $\vc{t} \in \mathbb{R}^d$
\Ensure The algorithm outputs a closest lattice vector $\vc{s} \in \cL$ to $\vc{t}$
\State Initialize $\vc{t}' \leftarrow \vc{t}$
\For{\textbf{each} $\vc{r} \in \cR$}
	\If{$\|\vc{t}' - \vc{r}\| < \|\vc{t}'\|$}
		\State Replace $\vc{t}' \leftarrow \vc{t}' - \vc{r}$ and restart the \textbf{for}-loop
	\EndIf
\EndFor
\State \Return $\vc{s} = \vc{t} - \vc{t}'$
\end{algorithmic}
\end{algorithm}

\section{The Doulgerakis--Laarhoven--De Weger randomized slicer}

As the iterative slicer of Sommer--Feder--Shalvi often does not succeed, when using as input only a subset of the relevant vectors of the lattice, Doulgerakis--Laarhoven--De Weger proposed the following heuristic variant of the slicer. Instead of using the list of relevant vectors for reductions, first we only use a subset of the relevant vectors. Since there is no guarantee that the slicer then returns a vector from the exact Voronoi cell, and the output may not be a solution, we repeat the algorithm many times on rerandomized versions of the same target vector. What this means is that instead of reducing $\vc{t}' = \vc{t}$ with the iterative slicer, we sample $\vc{t}' \sim \vc{t} + \cL$ at random (e.g.\ from a discrete Gaussian distribution over the coset $\vc{t} + \cL$) and repeat the algorithm on many such samples. This algorithm is given in pseudocode in Algorithm~\ref{alg:slicer3}.

In the worst case, each of these reductions will end up on the same path and reduce to the same, wrong solutions, thus making no progress. In practice however it was observed that, if the iterative slicer find a solution in a single run with probability $p \ll 1$, then repeating the algorithm $K$ times with such randomized target vectors leads to an overall success probability proportional to $K \times p$. This is purely an experimental, heuristic tweak -- there are no theoretical guarantees that reducing such a shifted target vector gives ``fresh'' results.

\begin{algorithm}[!ht]
\caption{The Doulgerakis--Laarhoven--De Weger randomized slicer~\cite{doulgerakis19}}
\label{alg:slicer3}
\begin{algorithmic}[1]
\Require A list $L \subset \cL$ and a target $\vc{t} \in \mathbb{R}^d$
\Ensure The algorithm outputs a closest lattice vector $\vc{s} \in \cL$ to $\vc{t}$
\State $\vc{s} \leftarrow \vc{0}$
\Repeat
	\State Sample $\vc{t}' \sim D_{\vc{t} + \cL,s}$
	\For{\textbf{each} $\vc{r} \in L$}
		\If{$\|\vc{t}' - \vc{r}\| < \|\vc{t}'\|$}
			\State Replace $\vc{t}' \leftarrow \vc{t}' - \vc{r}$ and restart the \textbf{for}-loop
		\EndIf
	\EndFor
	\If{$\|\vc{t}' - \vc{0}\| < \|\vc{t} - \vc{s}\|$}
		\State $\vc{s} \leftarrow \vc{t} - \vc{t}'$
	\EndIf
\Until{$\vc{s}$ is a closest lattice vector to $\vc{t}$}
\State \Return $\vc{s}$
\end{algorithmic}
\end{algorithm}

\section{The Doulgerakis--Laarhoven--De Weger complexity analysis}

To analyze the heuristic time and space complexities of the randomized slicer, Doulgerakis--Laarhoven--De Weger made the following assumptions. First, the vectors from the list $L \subset \cL$ are assumed to follow a spherically symmetric distribution, and their lengths are assumed to follow the prediction obtained via the Gaussian heuristic. Similarly, the exact Voronoi cell of the lattice is modeled as a ball of a certain radius, such that the volume of the ball matches the volume of the lattice. Containment of the reduced vector $\vc{t}' \in \vc{t} + \cL$ in $\cV$ was then estimated to be equivalent to the condition $\|\vc{t}'\| \leq \lambda_1(\cL)$.

Then, to analyze the success probability of the slicing routine, first it was observed that if $\vc{t}'$ has a rather large norm, then it is likely that $L$ contains a vector $\vc{v}$ such that $\vc{t}' - \vc{v}$ is shorter than $\vc{t}'$; progress can then still be made with ease. There is a phase transition at a certain value $\beta$ such that
\begin{itemize}
\item If $\|\vc{t}'\| > \beta$, then with probability at least $2^{-o(d)}$ there exists a vector $\vc{v} \in L$ such that $\|\vc{t}' - \vc{v}\| \leq \|\vc{t}'\|$;
\item If $\|\vc{t}'\| < \beta$, then with probability at most $2^{-\Theta(d)}$ there exists a vector $\vc{v} \in L$ such that $\|\vc{t}' - \vc{v}\| \leq \|\vc{t}'\|$.
\end{itemize}
After reaching norm $\beta$, the algorithm may still find a solution, but each additional reduction step is exponentially small to occur. To obtain a bound on the overall success probability of the algorithm, the authors studied the probability that after \emph{exactly one more reduction} with the list $L$, we reach the desired norm $\lambda_1(\cL)$, so that $\vc{t}'$ is expected to be contained in $\cV$. This is of course only one way for the algorithm to ``reach'' the Voronoi cell, and it may also happen that after two, three, and any number of additional reductions we still reach the solution, albeit with exponentially small probability. The analysis based on finding the solution in exactly one step, jumping from norm $\beta$ to $\lambda_1(\cL)$, is therefore only a lower bound on the overall success probability of the algorithm. This directly leads to the bound on the success probability stated in Equation~\eqref{eq:dlw111}.

Then, given this analysis of the algorithm, the authors obtained a lower bound on the success probability $p$ of a single run of the (randomized) iterative slicer. If one then makes the additional assumption that the success probability of the algorithm is \emph{equal} to the ratio of the volume of the exact cell over the volume of the approximate Voronoi cell, then this would immediately yield a lower bound on the ratio of these volumes as well. This would then lead to the conjectured lower bound on the ratio of the volumes given in Equation~\eqref{eq:dlw}.

As shown in this paper, the latter step is incorrect, as we give tight bounds on the ratio of these volumes, and show that the inverse of the expression from~\eqref{eq:dlw} is \emph{not} a lower bound on the ratio of the volumes.


\end{document}